\documentclass[12pt]{amsart}

\usepackage{amsfonts,amsmath,amsthm}
\usepackage{latexsym}

\newtheorem{theorem}{Theorem}[section]

\newtheorem{lemma}[theorem]{Lemma}

\newtheorem{problem}[theorem]{Problem}

\theoremstyle{definition}

\newcommand{\dst}{\displaystyle}

\newcommand \tr {\mathrm{Tr}\:}

\newcommand \sgn{\mathrm{sgn}\:}

\newcommand{\R}{\ensuremath{\mathbb{R}}}

\newcommand{\Co}{\ensuremath{\mathbb{C}}}

\newcommand{\one}{{\mathbf{1}}}
\newcommand{\inner}[2]{{\langle #1, #2 \rangle}}
\newcommand{\Abs}[1]{{\left|{#1}\right|}}

\def \e {\varepsilon}

\def \C {\mathbb{C}}

\newcommand{\ac}{\ensuremath{\mathcal{A}}}
\newcommand{\bc}{\ensuremath{\mathcal{B}}}

\newcommand{\eb}{\ensuremath{\mathbf{e}}}
\newcommand{\fb}{\ensuremath{\mathbf{f}}}

\def \< {\langle}
\def \> {\rangle}

\newcommand{\abs}[1]{{\left|{#1}\right|}}
\newcommand{\scal}[1]{{\left\langle{#1}\right\rangle}}

\begin{document}

\title[Positive definite functions and MUBs]{An application of positive definite functions to the problem of MUBs}

\author[M. N. Kolountzakis]{Mihail N. Kolountzakis}
\address{M. N. K.: Department  of  Mathematics  and  Applied  Mathematics,  University  of  Crete,  Voutes
Campus, 700 13 Heraklion, Greece.}
\email{kolount@gmail.com}

\author[M. Matolcsi]{M\'at\'e Matolcsi}
\address{M. M.: Budapest University of Technology and Economics (BME),
H-1111, Egry J. u. 1, Budapest, Hungary (also at Alfr\'ed R\'enyi Institute of Mathematics,
Hungarian Academy of Sciences, H-1053, Realtanoda u 13-15, Budapest, Hungary)}
\email{matomate@renyi.hu}

\author[M. Weiner]{Mih\'aly Weiner}
\address{M. W.: Budapest University of Technology and Economics (BME),
H-1111, Egry J. u. 1, Budapest, Hungary}
\email{mweiner@renyi.hu}

\thanks{M. Matolcsi was supported by the ERC-AdG 321104 and by NKFIH-OTKA Grant No. K104206, M. Weiner was supported by the ERC-AdG 669240 QUEST ``Quantum Algebraic Structures and Models'' and by NKFIH-OTKA Grant No. K104206}


\begin{abstract}
We present a new approach to the problem of mutually unbiased bases (MUBs), based on positive definite functions on the unitary group. The method provides a new proof of the fact that there are at most $d+1$ MUBs in $\Co^d$. It may also lead to a proof of non-existence of complete systems of MUBs in dimension 6 via a conjectured algebraic identity.
\end{abstract}

\maketitle

\bigskip

{\bf 2010 Mathematics Subject Classification.} Primary 15A30,
Secondary 43A35, 05B10

{\bf Keywords and phrases.} {\it  Mutually unbiased bases, positive definite functions, unitary group}

\section{Introduction}

In this paper we present a new approach to the problem of mutually unbiased bases (MUBs) in $\Co^d$. Our approach has been motivated by two recent results in the literature. First, in \cite{mubfourier} one of the present authors described how the Fourier analytic formulation of Delsarte's LP bound can be applied to the problem of MUBs. Second, in \cite[Theorem 2]{vallentinfilho} F. M. Oliveira Filho and F. Vallentin proved a general optimization bound which can be viewed as a generalization of Delsarte's LP bound to non-commutative settings (and they applied the theorem to packing problems in Euclidean spaces). As the MUB-problem is essentially a problem over the unitary group, it is natural to combine the two ideas above. Here we present another version of the non-commutative Delsarte scheme in the spirit of \cite[Lemma 2.1]{mubfourier}. Our formulation in Theorem \ref{ncdelsarte} below is somewhat less general than \cite[Theorem 2]{vallentinfilho}, but makes use of the underlying group structure and is very convenient for applications. It fits the MUB-problem naturally, and leads us to consider positive definite functions on the unitary group.

\bigskip

The paper is organized as follows. In the Introduction we recall some basic notions and results concerning mutually unbiased bases (MUBs). In Section \ref{sec2} we describe how the problem of MUBs fits into a non-commutative version of Delsarte's scheme. We then apply this method to give a new proof of the fact that there are at most $d+1$ MUBs in $\Co^d$. Finally, in Section \ref{sec3} we speculate on how the non-existence of complete systems of MUBs could be proved in dimension 6 via an algebraic identity conjectured in \cite{mubmols}.

\bigskip

Recall that two orthonormal bases in $\Co^d$,
$\ac=\{\eb_1,\ldots,\eb_d\}$ and $\bc=\{\fb_1,\ldots,\fb_d\}$ are
called \emph{unbiased} if for every $1\leq j,k\leq d$,
$\abs{\scal{\eb_j,\fb_k}}=\dst\frac{1}{\sqrt{d}}$. A collection
$\bc_1,\ldots\bc_m$ of orthonormal bases is said to be (pairwise)
\emph{mutually unbiased} if any two of them are
unbiased. What is the maximal number of mutually unbiased bases (MUBs) in $\Co^d$? This problem has its origins in quantum information theory, and has received considerable attention over the past decades (see e.g. \cite{durt} for a recent comprehensive survey on MUBs). The following upper bound is well-known (see e.g. \cite{BBRV,BBELTZ,WF}):

\begin{theorem}\label{thm1}
The number of mutually unbiased bases in $\Co^d$ is less than or equal to $d+1$.
\end{theorem}

We will give a new proof of this fact in Theorem \ref{thmgeneral} below. Another important result concerns the existence of complete systems of MUBs in prime-power dimensions (see e.g.
\cite{BBRV,Com1,Com2,Iv,KR,WF}).

\begin{theorem}\label{thm2}
A collection of $d+1$ mutually unbiased bases (called a {\it complete system} of MUBs) exists (and can be constructed explicitly)
if the dimension $d$ is a prime or a prime-power.
\end{theorem}

However, if the dimension $d=p_1^{\alpha_1}\dots p_k^{\alpha_k}$ is not a prime-power, very little
is known about the maximal number of MUBs.  By a tensor product construction it is easy to see that there are at least $p_j^{\alpha_j}+1$ MUBs in $\C^d$ where $p_j^{\alpha_j}$ is the smallest of the prime-power divisors of $d$. One could be tempted to conjecture the maximal number of MUBs always equals $p_j^{\alpha_j}+1$, but this is already known to be false: for some {\it specific} square dimensions $d=s^2$ a construction of \cite{262} yields more MUBs than $p_j^{\alpha_j}+1$ (the construction is based on orthogonal Latin squares). Another important phenomenon, proved in \cite{mweiner}, is that the maximal number of MUBs cannot be exactly $d$ (it is either $d+1$ or strictly less than $d$).

\medskip

The following basic problem remains open for all non-primepower dimensions:

\begin{problem}\label{MUB6problem}\
Does a complete system of $d+1$ mutually unbiased bases exist in $\Co^d$ if $d$ is not a prime-power?
\end{problem}

For $d=6$ it is widely believed among researchers that the answer is negative, and the maximal number of MUBs is 3. The proof still eludes us, however, despite considerable efforts over the past decade (\cite{BBELTZ,boykin,config,ujbrit,arxiv}). On the one hand, some {\it infinite families} of MUB-triplets
in $\C^6$ have been constructed (\cite{arxiv,Za}). On the other hand, numerical evidence strongly suggests that there exist no MUB-quartets \cite{config,ujbrit,numerical,Za}. For non-primepower dimensions other than 6 we are not aware of any conjectures as to the exact maximal number of MUBs.

\medskip

It will also be important to recall the relationship between mutually unbiased bases
and \emph{complex Hadamard matrices}. A $d\times d$ matrix $H$ is called a complex Hadamard matrix if all its entries have modulus 1 and $\frac{1}{\sqrt{d}}H$ is unitary. Given a collection of MUBs $\bc_1,\dots,\bc_m$ we may regard the bases as unitary matrices $U_1, \dots, U_m$ (with respect to some fixed orthonormal basis), and the condition of the bases being pairwise unbiased amounts to $U_i^\ast U_j$ being a complex Hadamard matrix scaled by a factor of $\frac{1}{\sqrt{d}}$ for all $i\ne j$. That is, $U_i^\ast U_j$ is a unitary matrix (which is of course automatic) whose entries are all of absolute value $\frac{1}{\sqrt{d}}$.

\medskip

A complete classification of MUBs up to dimension 5 (see \cite{BWB}) is based on the classification of complex Hadamard matrices (see \cite{haagerup}). However, the classification of complex Hadamard matrices in dimension 6 is still out of reach despite recent
efforts \cite{BN,k2,Msz,Skinner,generic}.

\medskip

In this paper we will use the above connection of MUBs to complex Hadamard matrices. In particular, we will describe a Delsarte scheme for non-commutative groups  in Theorem \ref{ncdelsarte}, and apply it to the MUB-problem with an appropriate witness function $h(Z)$ on the unitary group $U(d)$ in Theorem \ref{thmgeneral}.

\section{Mutually unbiased bases and a non-commutative Delsarte scheme}\label{sec2}

In this section we describe a non-commutative version of Delsarte's scheme, and show how the problem of mutually unbiased bases fit into this scheme. The commutative analogue was described in \cite{mubfourier}.

\medskip

Let $G$ be a compact group, the group operation being multiplication and the unit element being denoted by 1. We will denote the normalized Haar measure on $G$ by $\mu$. Let a symmetric subset $A=A^{-1}\subset G$, $1\in A$, be given. We think of $A$ as the 'forbidden' set. We would like to determine the maximal cardinality of a set $B=\{b_1, \dots b_m\}\subset G$ such that all the quotients $b_j^{-1}b_k\in A^c\cup \{1\}$ (in other words,  all quotients avoid the forbidden set $A$). When $G$ is commutative, some well-known examples of this general scheme are present in coding theory (\cite{delsarte}), sphere-packings (\cite{cohnelkies}), and sets avoiding square differences in number theory (\cite{ruzsa}). We will discuss the non-commutative case here.

\medskip

Recall that the convolution of $f, g\in L^1(G)$ is defined by $f\ast g (x)=\int f(y)g(y^{-1}x)d\mu(y)$

\medskip

Recall also the notion of positive definite functions on $G$. A function $h: G\to \C$ is called positive definite, if for any $m$ and any collection $u_1, \dots, u_m\in G$, and $c_1, \dots, c_m\in \C$ we have $\sum_{i,j=1}^m h(u_i^{-1}u_j)\overline{c_i}c_j\ge 0$. When $h$ is continuous, the following characterization is well-known.

\begin{lemma}\label{folland}{\rm (cf. \cite[Proposition 3.35]{folland})}
If $G$ is a compact group, and $h: G\to \C$ is a continuous function, the following are equivalent.

\smallskip

(i) $h$ is of positive type, i.e.
\begin{equation}\label{ptype}  \int (\tilde f \ast f)h\ge 0
\end{equation}
for all functions $f\in L^2(G)$ (here $\tilde f(x)=\overline{f(x^{-1})}$)

\smallskip

(ii) $h$ is positive definite
\end{lemma}

This statement is fully contained in the more general Proposition 3.35 in \cite{folland}. In fact, for compact groups Proposition 3.35 in \cite{folland} shows that instead of $L^2(G)$ the smaller class of continuous functions $C(G)$ or the wider class of absolute integrable functions $L^1(G)$ could also be taken in (i). All these cases are equivalent, but for us it will be convenient to use $L^2(G)$ in the sequel.

\medskip

We formulate another important property of positive definite functions.

\medskip

\begin{lemma}\label{posdeffunctions}
Let $G$ be a compact group and $\mu$ the normalized Haar measure on $G$. If $h: G\to \C$ is a continuous positive definite function then $\alpha=\int_G hd\mu\ge 0$, and for any $\alpha_0\le \alpha$ the function $h-\alpha_0$ is also positive definite. In other words, for any $m$ and any collection $u_1, \dots, u_m\in G$ and $c_1, \dots, c_m\in \C$ we have 
\begin{equation}\label{alpha}
\sum_{i,j=1}^m h(u_i^{-1}u_j)\overline{c_i}c_j\ge \alpha |\sum_{i=1}^{m}c_i|^2.
\end{equation}
\end{lemma}
\begin{proof}
Let $f\in L^2(G)$ and define a linear operator $H:L^2(G)\to L^2(G)$ by
\[
(Hf)(x) = \int h(x^{-1}y) f(y) \,dy.
\]
As $h$ is assumed to be positive definite, $H$ is positive self-adjoint. Also, writing $\one$ for the constant one function on $G$ we have
\[
H\one = \alpha \one,\ \ \ \inner{H\one}{\one} = \alpha\ge 0.
\]
Let us use the notation $\beta= \int f$. We have the orthogonal decomposition
\[
f = \beta\one + f_2,\ \ \ \mbox{where } f_2 \perp \one.
\]
Using invariance of the Haar measure and exchanging the order
of integrations one can easily find that
\[
\inner{Hf}{\one} = \int (Hf)(x)\overline{\one(x)} \, dx = 
\int h(x)\, dx \int f(y)\, dy = \alpha\beta.
\]
Note that here we have used the mathematician's convention according to which the scalar product is linear in its first, and conjugate linear in its second variable. Thus $\inner{Hf_2}{\one} = 0$, since
\[
\beta \alpha = \inner{Hf}{\one} = \inner{H(\beta\one + f_2)}{\one} = \beta\alpha + \inner{Hf_2}{\one}.
\]
To show that $h-\alpha$ is positive definite, we need to check 
that
\[
\inner{Hf}{f} - \Abs{\beta}^2 \alpha\ge 0
\]
for all $f\in L^2(G)$. We have 
\[
\inner{Hf}{f} = \inner{\beta\alpha \one + Hf_2}{\beta\one+f_2} = \Abs{\beta}^2 \alpha + \inner{Hf_2}{f_2}
\]
since $f_2 \perp \one$ and $Hf_2 \perp \one$. Hence $\inner{Hf}{f} - \Abs{\beta}^2 \alpha = \inner{Hf_2}{f_2} \ge 0$.
\end{proof}

\medskip

After these preliminaries we can describe the non-commutative analogue of Delsarte's LP bound. (To the best of our knowledge the commutative version was first introduced by Delsarte in connection with binary codes with prescribed Hamming distance \cite{delsarte}. Another formulation of the non-commutative version is given in \cite{vallentinfilho}).

\medskip

\begin{theorem}\label{ncdelsarte} (Non-commutative Delsarte scheme for compact groups)\\
Let $G$ be a compact group, $\mu$ the normalized Haar measure, and let $A=A^{-1}\subset G$, $1\in A$, be given. Assume that there exists a positive definite function $h: G\to \R$ such that $h(x)\le 0$ for all $x\in A^c$, and $\int hd\mu>0$. Then for any $B=\{b_1, \dots b_m\}\subset G$ such that $b_j^{-1}b_k\in A^c\cup \{1\}$ the cardinality of $B$ is bounded by $|B|\le \frac{h(1)}{\int hd\mu}$.
\end{theorem}
\begin{proof}
Consider
\begin{equation}\label{eq1}
S=\sum_{u,v\in B}h(u^{-1}v).
\end{equation}

On the one hand,
\begin{equation}\label{eq3}
S\le h(1)|B|,
\end{equation}
since all the terms $u\ne v$ are non-positive by assumption.

\medskip

On the other hand, applying \eqref{alpha} with $\alpha=\int h d\mu$, $u, v\in B$ and $c_u=c_v=1$, we get
\begin{equation}\label{eq2}
S\ge \alpha |B|^2.
\end{equation}

Comparing the two estimates \eqref{eq2}, \eqref{eq3} we obtain $|B|\le \frac{h(1)}{\int hd\mu}$.
\end{proof}

The function $h$ in the Theorem above is usually called a {\it witness function}.

\medskip

We will now describe how the problem of mutually unbiased bases fits into this scheme. Consider the group $U(d)$ of unitary matrices, being given with respect to some fixed orthonormal basis of $\C^d$. Consider the set $CH$ of complex Hadamard matrices. Following the notation of the Delsarte scheme above define $A^c=\frac{1}{\sqrt d} CH\subset U(d)$, i.e. let the {\it complement} of the forbidden set be the set of scaled complex Hadamard matrices. Then the maximal number of MUBs in $\C^d$ is exactly the maximal cardinality of a set $B=\{b_1, \dots b_m\}\subset U(d)$ such that all the quotients $b_j^{-1}b_k\in A^c\cup \{1\}$. After finding an appropriate witness function we can now give a new proof of the fact the number of MUBs in $\C^d$ cannot exceed $d+1$.

\begin{theorem}\label{thmgeneral}
The function $h(Z)=-1+\sum_{i,j=1}^d|z_{i,j}|^4$ (where $Z=(z_{i,j})_{i,j=1}^d\in U(d)$) is positive definite on $U(d)$, with $h(1)=d-1$ and $\int h=\frac{d-1}{d+1}$. Consequently, the number of MUBs in dimension $d$ cannot exceed $d+1$.
\end{theorem}
\begin{proof}
Consider the function $h_0(Z)=\sum_{i,j=1}^d|z_{i,j}|^4$. First we prove that $h_0$ is positive definite. For this, recall that the Hilbert-Schmidt inner product of matrices is defined as $\inner{X}{Y}_{HS}=\tr (XY^\ast)$, and for any vector $v$ in a finite dimensional Hilbert space $H$ the (scaled) projection operator $P_v$ is defined as $P_v u=\inner{u}{v}v$. For any two vectors $u,v\in H$ we have $|\inner{u}{v}|^2=\tr P_uP_v$. Also, recall that the inner product on $H\otimes H$ is given by $\inner{u_1\otimes u_2}{v_1\otimes v_2}=\inner{u_1}{v_1}\inner{u_2}{v_2}$.

\medskip

Let $U_1, \dots, U_m$ be unitary matrices, $c_1, \dots, c_m\in \C$, and let $\{e_1, \dots, e_d\}$ be the orthonormal basis with respect to which the matrices in $U(d)$ are given. Then

\begin{equation}\label{h0}
|\langle U_r^\ast U_t e_j, e_k \rangle|^4=|\langle U_te_j, U_re_k\rangle|^4=|\inner{U_te_j\otimes U_te_j}{U_re_k\otimes U_re_k}|^2=
\end{equation}
\begin{equation*}
\tr P_{U_te_j\otimes U_te_j}P_{U_re_k\otimes U_re_k}.
\end{equation*}

Therefore, with the notation $Q_t=\sum_{j=1}^m P_{U_te_j\otimes U_te_j}$ we have

\begin{equation}
h(U_r^\ast U_t)=\sum_{j,k}|\langle U_r^\ast U_t e_j, e_k \rangle|^4= \tr Q_tQ_r.
\end{equation}

Finally,

\begin{equation}
\sum_{r,t=1}^m h(U_r^\ast U_t)\overline{c_r}c_t=\|\sum_{t=1}^m c_tQ_t\|_{HS}^2\ge 0,
\end{equation}
as desired.
\medskip

It is known \cite{integral} that the integral of $h_0$ on $U(d)$ is $\frac{2d}{d+1}$. By applying Lemma \ref{posdeffunctions} to $h_0$ with $\alpha_0=1<\int h_0$ we get that $h$ is also positive definite. Note also that $h$ vanishes on the set $\frac{1}{\sqrt d} CH$ of scaled complex Hadamard matrices, $h(1)=d-1$, and $\int h=\frac{2d}{d+1}-1=\frac{d-1}{d+1}$. Therefore, Theorem \ref{ncdelsarte} implies that the number of MUBs in $\C^d$ is less than or equal to $\frac{h(1)}{\int h}=d+1$.
\end{proof}

\medskip

We remark here that one could consider the witness functions $h_\beta=h_0-\beta$ for any $1\le \beta \le \frac{2d}{d+1}$. All these functions satisfy the conditions of Theorem \ref{ncdelsarte}. However, an easy calculation shows that the best bound is achieved for $\beta=1$.

\section{Dimension 6}\label{sec3}

In particular, let us examine the situation in dimension $d=6$.

\medskip

The function $h(Z)-1+\sum_{i,j=1}^d|z_{i,j}|^4$ in Theorem \ref{thmgeneral} was a fairly natural candidate which vanishes on the set of (scaled) complex Hadamard matrices $\frac{1}{\sqrt d}CH$, for any $d$. However, for $d=6$ we have other functions which are conjectured to vanish on $\frac{1}{\sqrt d}CH$. Namely, Conjecture 2.3 in \cite{mubmols} provides a selection of such functions. Let

\begin{equation}\label{m1}
m_1(Z)=\sum_{\pi \in S_6} \sum_{j=1}^{6}z_{\pi(1), j}z_{\pi(2), j}z_{\pi(3), j}\overline{z_{\pi(4), j}}\overline{z_{\pi(5), j}}\overline{z_{\pi(6), j}},
\end{equation}
where $S_6$ denotes the permutation group on 6 elements. Also, let $m_2(Z)=m_1(Z^\ast)$. Then $m_1$ and $m_2$ are real-valued (because each term appears with its conjugate), and they are conjectured to vanish on $\frac{1}{\sqrt d}CH$. This provides some natural candidates for witness functions for the MUB-problem. Namely, let $m(Z)=(m_1(Z)+m_2(Z))^2$, or $m(Z)=m_1^2(Z)+m_2^2(Z)$, or $m(Z)=(m_1(Z)m_2(Z))^2$. In all three cases $m(I)=0$, and $\int_{Z\in U(d)} m(Z)d\mu >0$. Therefore, if for any $\e>0$ the function $h(Z)+\e m(Z)$ is positive definite, we get a better bound than in Theorem \ref{thmgeneral}, and obtain that the number of MUBs in dimension 6 is strictly less than 7, i.e. a complete system of MUBs does not exist. Unfortunately, we have not yet been able to prove that such $\e>0$ exists for any of the functions $m(Z)$ above.

\medskip

Furthermore, as the inner sum in \eqref{m1} is conjectured to be zero for all $\pi\in S_6$, we may even multiply each term with $(-1)^{\sgn\pi}$, if we wish. This leads to other possible choices of $m(z)$.

\medskip

It would also be interesting to find any analogue of Conjecture 2.3 in \cite{mubmols} for any dimensions other than $d=6$.

\end{document}